\documentclass[11pt]{article}

\usepackage{bm}
\usepackage{epsf,amsfonts,amsmath}
\usepackage{amsmath}
\usepackage{amssymb}
\usepackage{amsthm}
\newtheorem{thm}{Theorem}[subsection]
\usepackage{epic,eepic}
\usepackage{texdraw}
\usepackage[dvips]{color}
\usepackage{color}
\usepackage{graphicx}
\usepackage{exscale,relsize}
\usepackage{authblk}
\usepackage{url}
\usepackage{enumerate}
\allowdisplaybreaks
\newcommand{\field}[1]{\mathbb{#1}}
\newtheorem{proposition}[thm]{Proposition}
\newtheorem{theorem}[thm]{Theorem}

\newtheorem{remark}[thm]{Remark}
\DeclareMathOperator{\tr}{tr}
\DeclareMathOperator{\id}{id}

\DeclareMathOperator{\End}{End}

\title{\textbf{Entwining Yang-Baxter maps related to \\NLS type equations}}

\author{S. Konstantinou-Rizos\thanks{s.konstantinu.rizos@uniyar.ac.ru, skonstantin84@gmail.com.}}
\author[2]{G. Papamikos\thanks{G.Papamikos@leeds.ac.uk, geopap1983@gmail.com.}}
\affil{Centre of Integrable Systems, P.G. Demidov Yaroslavl State 
University, Russia}
\affil[2]{School of Mathematics, University of Leeds, UK}

\begin{document}

\maketitle
\begin{abstract}
We construct birational maps that satisfy the parametric set-theoretical Yang-Baxter equation and its entwining generalisation. For this purpose, we employ Darboux transformations related to integrable Nonlinear Schr\"odinger type equations and study the refactorisation problems of the product of their associated Darboux matrices. Additionally, we study various algebraic properties of the derived maps, such as invariants and associated symplectic or Poisson structures, and we prove their complete integrability in the Liouville sense. 
\end{abstract}

\hspace{.2cm} \textbf{PACS numbers:} 02.30.Ik, 02.90.+p, 03.65.Fd.

\hspace{.2cm} \textbf{Mathematics Subject Classification:} 35Q53, 81R12.

\hspace{.2cm} \textbf{Keywords:} Entwining parametric Yang-Baxter maps, Darboux transformations, Liouville  

\hspace{2.4cm} integrability, NLS type equations.

\section{Introduction}
This paper is concerned with the construction and the study of integrability of birational solutions to the parametric entwining  Yang-Baxter equation \cite{Brzezinski, KouloukasBanach}
\begin{equation}\label{YB_Ent}
S^{12}\circ R^{13} \circ T^{23}=T^{23}\circ R^{13} \circ S^{12}.
\end{equation}
When $S\equiv R\equiv T$, equation \eqref{YB_Ent} becomes the classical, set-theoretical Yang-Baxter equation \cite{Drinfel'd} which has appeared in many fields of mathematics and physics, see \cite{Veselov3} for a review, and its importance is difficult to overstate. During the last 2-3 decades a lot of interest has been drawn to the study and applications of solutions to the set-theoretical Yang-Baxter equation (the so-called Yang-Baxter maps \cite{Veselov}), and its entwining version \eqref{YB_Ent}; indicatively, we refer to \cite{ABS-2005,  Vincent, Allan-Pavlos, Pavlos, Pavlos2019, Sokor-Sasha, Kouloukas, Kouloukas2, PT, PTV, PSTV, Veselov2, Veselov, Veselov3}. In the modern theory of integrable systems there is a connection between the 3D-consistency property, which is recognised to be a very important property of discrete integrable systems, and the Yang-Baxter equation \cite{ABS-2005, Caudrelier, Pavlos, PT, PTV}. Moreover, there exist solutions which can be associated to matrix refactorisation problems and can be further related to Darboux and B\"acklund transformations \cite{Sokor-Sasha} of soliton equations.

Due to the significance of equation \eqref{YB_Ent}, its connection with various fields of mathematical and physics and its numerous applications in the theory of Integrable Systems, it is important to construct new solutions and study their properties. The motivation to study solutions of the Yang-Baxter equation related to Darboux transformations arises naturally from some recently obtained results \cite{Sokor-Sasha, MPW} where Yang-Baxter maps were obtained using Darboux transformations related to several integrable models. In particular, in \cite{Sokor-Sasha}, Darboux transformations associated with the nonlinear Schr\"odinger equation (NLS), derivative NLS (DNLS) and deformation of the DNLS (dDNLS) equation were employed in order to construct birational integrable Yang-Baxter maps, while in \cite{MPW} a Yang-Baxter map on the sphere constructed in relation to a vectorial generalisation of the sine-Gordon equation. 
The Lax operators of the DNL and the dDNLS equation are invariant under the action of reduction groups isomorphic to $\mathbb{Z}_2$ and $\mathbb{D}_2$, with rational dependence in the spectral parameter and having poles in $\overline{\mathbb{C}}=\mathbb{C}\cup \lbrace\infty \rbrace$ at generic and degenerate (nongeneric) orbits; in \cite{Sokor-Sasha}, only the operators corresponding to degenerate orbits were used. In the first part of this paper, we study the refactorisation problem of the product of two Darboux matrices associated to the DNLS Lax operator, invariant under the $\mathbb{Z}_2$ reduction group with poles in $\overline{\mathbb{C}}$ that form a single generic orbit under the action of the reduction group. In the second part, we employ different Darboux transformations related to the NLS operator and the DNLS operator, invariant under $\mathbb{Z}_2$ having poles at degenerate points and we construct solutions to the entwining Yang-Baxter equation.

The paper is organised as follows. In the next section, we provide preliminary material, definitions and we fix the notation. In particular, we give the definition of the entwining Yang-Baxter equation and we explain the relation between entwining Yang-Baxter maps and matrix refactorisation problems. For entwining maps defined by matrix refactorisation problems, we discuss their birationality, we prove a statement regarding and their invariants and we give the definition of their integrability in the Liouville sense. Finally, we list all the Darboux matrices we are using throughout the text. In section 3, we employ the Darboux transformation for the DNLS Lax operator invariant under the $\mathbb{Z}_2$ reduction group to construct a completely integrable involutive Yang-Baxter map, which can be restricted to the well-known $H_{IV}$ map on symplectic leaves. Moreover, we derive a non-involutive six-dimensional Yang-Baxter map which can be restricted to a three-dimensional map on invariant leaves. Section 4 deals with the derivation of new, birational, parametric, entwining Yang-Baxter maps together with their Lax representations. Again, the maps are derived using Darboux transformations related to the NLS and the DNLS equation, by considering matrix refactorisation problems with different Lax matrices. 
We prove complete integrability of the obtained maps. Finally, in Section 5, we give a summary of the obtained results and explore some ideas for future work.


\section{Preliminaries}
\subsection{Parametric entwining Yang-Baxter equation and matrix refactorisation problems}
Let $V$ be an algebraic variety in $K^N$, where $K$ is any field of characteristic zero; here, $K=\mathbb{C}$ Let also $S_{a,b}\in\End(V\times V)$, $R_{a,b}\in\End(V\times V)$ and $T_{a,b}\in\End(V\times V)$ be three maps defined by
\begin{subequations}
\begin{align}
&(x,y) \stackrel{S_{a,b}}{\longmapsto}(u_s(x,y;a,b),v_s(x,y;a,b)),\\
&(x,y)\stackrel{R_{a,b}}{\longmapsto}(u_r(x,y;a,b),v_r(x,y;a,b)),\\
&(x,y) \stackrel{T_{a,b}}{\longmapsto}(u_t(x,y;a,b),v_t(x,y;a,b)),
\end{align}
\end{subequations}
where $x,y,u,v\in V$, $a,b\in\mathbb{C}$. The \textit{parametric entwining Yang-Baxter equation} reads 
\begin{equation}\label{YB_Ent-ab}
S^{12}_{a,b}\circ R^{13}_{a,c} \circ T^{23}_{b,c}=T^{23}_{b,c}\circ R^{13}_{a,c} \circ S^{12}_{a,b},
\end{equation}
where the upper index $ij$ denotes that the corresponding map acts on $i$ and $j$ component of the triple $(x,y,z)\in V\times V\times V$ according to its definition, and identically on the third remaining. More precisely,
\begin{equation*}
S^{12}_{a,b}(x,y,z)=S_{a,b}\times\id,\quad R^{23}_{a,c}=\id\times R_{a,c} \quad\text{and}\quad T^{13}_{b,c}=\pi^{12}\times T^{23}_{b,c}\times\pi^{12},
\end{equation*}
where $\pi^{12}$ is the permutation map in $V\times V\times V$ defined by $\pi^{12}(x,y,z)=(y,x,z)$. \textit{Parametric entwining Yang-Baxter map} is a system ($S_{a,b},R_{a,b},T_{a,b}$) of maps satisfying \eqref{YB_Ent-ab}. Schematically, the latter equation can be represented on the edges of the cube.
\begin{figure}[ht]
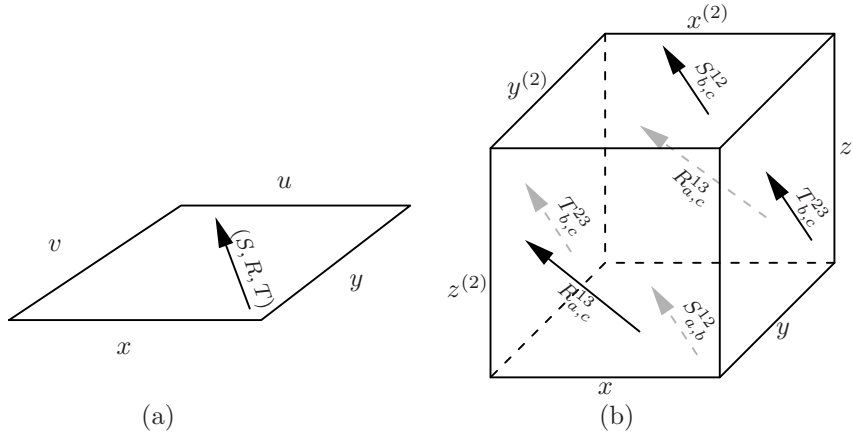

\centering
\centertexdraw{ 
\setunitscale 0.6
\move (-2.2 0.5)  \lvec(0 0.5) \lvec(1.3 1.5) \lvec(-0.7 1.5) \lvec(-2.2 0.5)
\textref h:C v:C \small{\htext(-1.2 0.25){$x$}}
\textref h:C v:C \htext(0.2 1.7){$u$} 
\textref h:C v:C \htext(0.83 0.82){$y$}
\textref h:C v:C \htext(-1.8 1.15){$v$}  
\textref h:C v:C \scriptsize{\rtext td:-70 (-.08 0.95){$(S,R,T)$}}
\move (-0.1 0.6) \arrowheadtype t:F \avec(-0.4 1.4)
\move (2.7 1.1) \arrowheadtype t:F \lpatt(0.067 0.1)  
\setgray 0.7
\avec(2.3 1.7)
\textref h:C v:C \scriptsize{\rtext td:-55 (2.7 1.4){$T^{23}_{b,c}$}}
\move (4.4 1.4) \arrowheadtype t:F \footnotesize{\avec(3.3 2.2)}
\textref h:C v:C { \scriptsize{ \rtext td:-40 (3.75 1.66){\setgray 0.7 $R^{13}_{a,c}$}}}
\move (3.8 0.2) \arrowheadtype t:F \avec(3.4 0.8)
\textref h:C v:C \scriptsize{\rtext td:-55 (3.8 0.5){$S^{12}_{a,b}$}}
\lpatt()
\setgray 0
\move (2 0)  \lvec(4 0) \lvec(5 1) 
\lpatt(0.067 0.1) \lvec(3 1) \lvec(2 0) 
\lpatt() \lvec(2 2) \lvec(3 3) 
\lpatt (0.067 0.1) \lvec(3 1) 
\lpatt() \move (3 3) \lvec(5 3) \lvec(4 2) \lvec(4 0) 
\move (2 2) \lvec(4 2)
\move (5 3) \lvec(5 1)
\textref h:C v:C \small{\htext(3 -0.1){$x$}}
\textref h:C v:C \small{\htext(4.55 0.4){$y$}}
\textref h:C v:C \small{\htext(5.1 2){$z$}}
\move (4.8 1.2) \arrowheadtype t:F \avec(4.4 1.8)
\textref h:C v:C \scriptsize{\rtext td:-55 (4.78 1.48){$T^{23}_{b,c}$}}
\move (3.3 0.4) \arrowheadtype t:F \avec(2.3 1.2)
\textref h:C v:C \scriptsize{\rtext td:-40 (2.75 0.66){$R^{13}_{a,c}$}}
\textref h:C v:C \small{\htext(1.8 0.8){$z^{(2)}$}}
\move (3.9 2.3) \arrowheadtype t:F \avec(3.5 2.9)
\textref h:C v:C \scriptsize{\rtext td:-55 (3.9 2.6){$S^{12}_{b,c}$}}
\textref h:C v:C \small{\htext(3.9 3.15){$x^{(2)}$}} \htext(2.33 2.55){$y^{(2)}$}
\textref h:C v:C \htext(3.1 -0.35){(b)}
\textref h:C v:C \htext(-0.9 -0.35){(a)}
}
\caption{Cubic representation of (a) the parametric YB map\index{Yang-Baxter (YB) map(s)!parametric} and (b) the corresponding YB equation.}\label{YBmap-Eq}\index{Yang-Baxter (YB) equation}
\end{figure}

\begin{remark}\normalfont
The maps $S_{a,b}$, $R_{a,b}$ and $T_{a,b}$ individually do not necessary satisfy the Yang-Baxter equation.
\end{remark}

The relation between equation \eqref{YB_Ent-ab} and matrix refactorisation problems was studied in \cite{KouloukasBanach}, following \cite{Veselov2}. Specifically, a \textit{Lax triple} for a set of maps ($S_{a,b},R_{a,b},T_{a,b}$) is a set of operators ($L_1,L_2,L_3$), $L_i=L_i(x;a,\lambda)$, $i=1,2,3$, which depend on a variable $x\in V$, a parameter $a\in\mathbb{C}$ and a spectral parameter\footnote{We will not be writing explicitly the dependence on the spectral parameter.} $\lambda\in\mathbb{C}$, such that
\begin{subequations}\label{L123}
\begin{align}
&L_1(u_s;a)L_2(v_s;b)=L_2(y;b)L_1(x;a);\label{L12}\\
&L_1(u_r;a)L_3(v_r;b)=L_3(y;b)L_1(x;a);\label{L13}\\
&L_2(u_t;a)L_3(v_t;b)=L_3(y;b)L_2(x;a).\label{L23}
\end{align}
\end{subequations}
If \eqref{L12}, \eqref{L13} and \eqref{L23} are equivalent to $(u_s,v_s)=S_{a,b}(x,y)$, $(u_r,v_r)=R_{a,b}(x,y)$ and $(u_t,v_t)=T_{a,b}(x,y)$, respectively, then the triple ($L_1,L_2,L_3$) is called \textit{strong} Lax triple \cite{KouloukasBanach}. For $L_1=L_2=L_3=L$ and $S_{a,b}=R_{a,b}=T_{a,b}$, the above equations reduce to 
\begin{equation}\label{Lax-eq}
L(u;a)L(v;b)=L(y;b)L(x;a)
\end{equation}
which is the classical Lax representation of a Yang-Baxter map \cite{Veselov2}. If \eqref{Lax-eq} defines a rational map, then this map is birational (see \cite{Sokor-Sasha}).

\begin{remark}\normalfont
Rational entwining Yang-Baxter maps defined by matrix refactorisation problems \eqref{L123} are not necessarily birational. This is due to a symmetry break; that is, unlike \eqref{Lax-eq}, equation \eqref{L12}, for instance, is not invariant under the transformation $(u_s,v_s;a,b)\leftrightarrow (y,x;b,a)$. It turns out that all the obtained maps in this paper are birational; however, birationality is verified for every map individually. 
\end{remark}

The `entwining' Yang-Baxter property of the maps defined from the matrix refactorisation problems \eqref{L123} can be verified by straight forward substitution to \eqref{YB_Ent-ab}. However, since the associated maps are rational, and the triple composition of rational maps is usually complicated, we will be using the following criterion which allows one to work only with polynomial expressions. 
\begin{proposition} (Kouloukas-Papageorgiou \cite{KouloukasBanach})
If the matrix trifactorisation $L_1(x;a)L_2(y;b) L_3(z;c)$ $=L_1(x';a)L_2(y';b)L_3(z';c)$ implies $x=x',y=y',z=z'$, then the system of maps ($S_{a,b},R_{a,b},T_{a,b}$), defined by the refactorisation problems \eqref{L123}, is a parametric entwining Yang-Baxter map.
\end{proposition}

Regarding the invariants of parametric entwining Yang-Baxter maps, we have the following.

\begin{proposition}
The quantities $\tr(L_2(y;b)L_1(x;a))$, $\tr(L_3(y;b)L_1(x;a))$ and $\tr(L_3(y;b)L_2(x;a))$ constitute generating functions of invariants for maps $S_{a,b}$, $R_{a,b}$ and $T_{a,b}$, respectively.
\end{proposition}
\begin{proof}
We have
\begin{equation}\label{inv-gen}
\tr(L_1(u_s;a)L_2(v_s;b))\stackrel{\eqref{L12}}{=}\tr(L_2(y;b)L_1(x;a))=\tr(L_1(x;a)L_2(y;b)).
\end{equation}
If we expand function $\tr(L_1(x;a)L_2(y;b))=\sum_k I_k(x,y;a,b)\lambda^k$, it follows from \eqref{inv-gen} that
$$
I_i(u,v;a,b)=I_i(x,y;a,b),
$$
i.e. $I_i$ are invariants  of map $S_{a,b}$. Similarly we prove for $R_{a,b}$ and $T_{a,b}$.
\end{proof}

\subsection{Liouville integrability}
In this paper we study the integrability of the derived Yang-Baxter and entwining Yang-Baxter maps in the Liouville sense. Following \cite{Maeda, Veselov4}, an $2N$-dimensional map is said to be \textit{completely} or \textit{Liouville} integrable, if it admits $N$ functionally independent and globally defined invariants, $I_i$, $i=1,\ldots,N$, that are in involution with respect to an invariant Poisson bracket, i.e. $\{I_i,I_j\}=0$, $i\neq j$ and the map is Poisson with respect to the bracket.

\subsection{Darboux transformations}
In this paper, we use Darboux transformations related to the Lax operators $\mathcal{L}_{NLS} \in\mathfrak{sl}_2\left[\lambda\right]\left[D_x\right]$ and $\mathcal{L}_{DNLS}\in\mathfrak{sl}_2[\lambda]^{\langle s_1 \rangle}[D_x]$, namely
\begin{equation*}
\mathcal{L}_{NLS} =D_x+\lambda \left(\begin{array}{cc}1 & 0\\0 & -1\end{array}\right)+\left(\begin{array}{cc}0 & 2p\\2q & 0\end{array}\right),\qquad
\mathcal{L}_{DNLS}=D_x+\lambda^2 \left(\begin{array}{cc}1 & 0\\0 & -1\end{array}\right)+\lambda\left(\begin{array}{cc}0 & 2p\\2q & 0\end{array}\right)
\end{equation*}
i.e. the spatial parts of the Lax pairs for the nonlinear Schr\"odinger (NLS) and the derivative nonlinear Schr\"odinger (DNLS) equation. The second operator is invariant under the transformation $s_1: \mathcal{L}(\lambda) \rightarrow \sigma_{3}\mathcal{L}(-\lambda) \sigma_{3}$; this involution generates the reduction group which is isomorphic to the $\field{Z}_2$ group. The automorphism $\lambda\rightarrow -\lambda$ has both degenerate and generic orbits in $\overline{\mathbb{C}}$. Generic orbits are those orbits of a group action that correspond to a generic point of the space, while degenerate orbits are those which correspond to a nongeneric (fixed and periodic point).


\subsubsection{NLS case}
The well-known Darboux matrix for the NLS equation\footnote{Here by $(p,q)$ we denote the original potentials, whereas by $(\tilde{p},\tilde{q})$ we denote the new potentials obtained by the Darboux map.} is 
\begin{equation}\label{M-NLS}
M(p,\tilde{q}):=\lambda \left(\begin{array}{cc}1 & 0\\0 & 0\end{array}\right)+\left(\begin{array}{cc}a+p\tilde{q} & p\\\tilde{q} & 1 \end{array}\right),
\end{equation}
and its degeneration \cite{SPS} is given by
\begin{equation}\label{M-degen}
M_c(p,f)=\lambda \left(\begin{array}{cc}1 & 0\\0 & 0\end{array}\right)+\left(\begin{array}{cc}f & p\\\frac{c}{p} & 0 \end{array}\right),\quad f=\frac{p_x}{2p}.
\end{equation}
\subsubsection{DNLS case: Degenerate orbits}
In the case of the DNLS equation, a Darboux transformation is given by (\cite{SPS}) the following matrix $M$
\begin{equation} \label{DT-sl2-gen}
M(p,\tilde{q},f;c_1,c_2) := \lambda^{2}\left(\begin{array}{cc} f & 0\\ 0 & 0\end{array}\right)+\lambda\left(\begin{array}{cc} 0 & fp\\ f\tilde{q} & 0\end{array}\right)+\left(\begin{array}{cc} c_1 & 0\\ 0 & c_2 \end{array}\right),
\end{equation}
where $p$ and $q$ satisfy a system of differential equations (\cite{Sokor-Sasha}) which possesses the following first integral,
\begin{equation} \label{sl2-D-con-det-gen}
\partial_x \left(f^{2}p\tilde{q}-c_2 f\right)=0.
\end{equation}
For $c_1=c_2=1$, using the above first integral, one can replace $f$ by
\begin{equation} \label{sl2-D-int}
f=k+(fp)(f\tilde{q}).
\end{equation}
In the case where $c_1=1$ and $c_2=0$, after eliminating $\tilde{q}$ with use of (\ref{sl2-D-con-det-gen}), the Darboux matrix (\ref{DT-sl2-gen}) degenerates to
\begin{equation} \label{DT-sl2-degen}
M(p,f;k) := \lambda^{2}\left(\begin{array}{cc} f & 0\\ 0 & 0\end{array}\right)+\lambda\left(\begin{array}{cc} 0 & fp\\ \frac{k}{fp} & 0\end{array}\right)+\left(\begin{array}{cc} 1 & 0\\ 0 & 0 \end{array}\right).
\end{equation}

\subsubsection{DNLS case: Generic orbits}
\begin{equation} \label{Z2-M-D1}
M(p, \tilde{q})\, =\, \frac{1}{\lambda-1} \left(\begin{array}{cc} \tilde{q} & 1\\1
& -p \end{array} \right) - \frac{1}{\lambda+1}\left(\begin{array}{cc} \tilde{q} &
-1\\-1 & -p \end{array} \right), \end{equation}
where
\begin{equation*}
p\tilde{q} = -1, \quad \partial_x p=4 p \left(\frac{p}{p-q}-\frac{\tilde{p}}{\tilde{p}-\tilde{q}} \right).
\end{equation*}

Another Darboux matrix is given by
\begin{equation} \label{Z2-M-gen}
M(p, \tilde{q}, f;c_1, c_2) = \frac{f}{\lambda-1} \left(\begin{array}{cc}
\tilde{q} & -p \tilde{q} \\ 1 & -p \end{array} \right) -
\frac{f}{\lambda+1}\left(\begin{array}{cc} \tilde{q} & p \tilde{q} \\ -1 & -p
\end{array} \right) + \left(\begin{array}{cc} c_1 &0 \\ 0 & c_2 \end{array}
\right), 
\end{equation}
where $p$,  $\tilde{q}$ and $f$ satisfy a system of differential equations (see \cite{Sokor-Sasha}) which possess the following first integral
\begin{equation} \label{z2-fi}
\Phi(c_1, c_2) = \left(2 f p+c_2\right) \left(2 f \tilde{q}-c_1\right), \qquad |c_1|^2+|c_2|^2 \ne 0.
\end{equation}

\section{Derivation of Yang-Baxter maps related to the DNLS equation}
In this section we make use of Darboux matrices associated to the DNLS equation, in the case of $\mathbb{Z}_2$ reduction with generic orbits, to derive solutions to the classical Yang-Baxter equation. We study the integrability of the obtained maps.

\subsection{An involutive Yang-Baxter map}
We change $(p,\tilde{q})\rightarrow (x_1,x_2)$ in \eqref{Z2-M-D1}, namely we consider the following matrix
\begin{equation} 
M(\textbf{x}) = \frac{1}{\lambda-1} \left(\begin{array}{cc} x_2 & 1\\1
& -x_1 \end{array} \right) - \frac{1}{\lambda+1}\left(\begin{array}{cc} x_2 &
-1\\-1 & -x_1 \end{array} \right), \quad \textbf{x}:=(x_1,x_2),
\end{equation}
and substitute it to the Lax equation
\begin{equation}\label{L-DNLS_gen}
M(\textbf{u})M(\textbf{v})=M(\textbf{y})M(\textbf{x}).
\end{equation}
The above, Lax equation will imply a set of polynomial equations in variables $\textbf{u},\textbf{v},\textbf{x},\textbf{y}$ which define the algebraic variety. In this case, the algebraic variety is union of two four-dimensional components. The first one is obvious from refactorisation \eqref{L-DNLS_gen} and it corresponds to the permutation map
\begin{equation}
\textbf{u}\mapsto \textbf{y}, \qquad \textbf{v} \mapsto \textbf{x},
\end{equation}
which is a trivial YB map. The second one can be represented as the following involutive, birational, four-dimensional YB map
\begin{equation}\label{YB-DNLS_gen}
(\textbf{x},\textbf{y})\overset{Y}{\longrightarrow } (\textbf{u},\textbf{v})=\left(\frac{y_1-x_2}{x_1-y_2}x_1,\frac{y_2-x_1}{x_2-y_1}x_2,\frac{y_2-x_1}{x_2-y_1}y_1,\frac{y_1-x_2}{x_1-y_2}y_2 \right).
\end{equation}

For the integrability of this map we have the following.

\begin{proposition}
Map \eqref{YB-DNLS_gen} is completely integrable.
\end{proposition}
\begin{proof}
From the trace of the monodromy matrix we obtain the following invariant
\begin{equation}
I=x_1y_1+x_2y_2.
\end{equation}
However, this is a sum of the following two invariants
\begin{equation}\label{2invs}
I_i=x_iy_i, \qquad i=1,2.
\end{equation}
These are in involution with respect to the following Poisson bracket
\begin{equation}
\{x_1,x_2\}_1=-x_1x_2, \quad\{y_1,y_2\}_1=y_1y_2\quad\{x_1,y_2\}_1=x_1y_2 \quad\text{and}\quad\{x_i,y_j\}_1=x_iy_j,~~i,j=1,2.
\end{equation}
\end{proof}

\begin{remark}\normalfont
Another Poisson bracket is
\begin{subequations}
\begin{align}
&\{x_1,x_2\}_2=-x_1x_2, \quad \{y_1,y_2\}_2=-y_1y_2,\quad\{x_1,y_2\}_2=2x_1y_2  \\
&\{x_i,y_i\}_2=-x_iy_i, ~~i=1,2, \quad \text{and}\quad \{x_2,y_1\}=0.
\end{align}
\end{subequations}
It is  $\{I_1,I_2\}_2=0$. Here, the rank of the corresponding Poisson matrix is 2, and the following quantities are Casimir functions 
\begin{equation}
C_1=\frac{x_2}{y_1} \quad \text{and} \quad C_2=\frac{y_1}{x_2}.
\end{equation}
The latter are preserved by \eqref{YB-DNLS_gen}, namely $C_i\circ Y=C_i$, $i=1,2$. Therefore, map \eqref{YB-DNLS_gen} is completely integrable.
\end{remark}


\subsection{Restriction to $H_{IV}$ on symplectic leaves}
We have the following.

\begin{proposition}
\begin{enumerate}
	\item Map \eqref{YB-DNLS_gen} can be restricted to a two-dimensional YB map, $Y\in\End(A_a\times B_b)$ on the symplectic leaves $A_a:=\left\{(x_1,y_1)\in K^2;x_1y_1=a\right\}$ and $B_b:=\left\{(x_2,y_2)\in K^2;x_2y_2=b\right\}$;
	\item This two-dimensional map is the $H_{IV}$ map.
\end{enumerate}
\end{proposition}

\begin{proof}
\begin{enumerate}
	\item The restriction can be achieved using invariants \eqref{2invs}. In particular, setting $x_1y_1=a$ and $x_2y_2=b$, we can eliminate $y_1$ and $y_2$ from \eqref{YB-DNLS_gen}. The resulting map $x_1\mapsto u_1(x_1,x_2)$, $x_2\mapsto u_2(x_1,x_2)$ is given by
	\begin{equation}
	(x_1,x_2)\overset{Y_{a,b}}{\longrightarrow } (u_1(x_1,x_2),u_2(x_1,x_2))=\left(\frac{a-x_1x_2}{-b+x_1x_2}x_2,\frac{b-x_1x_2}{-a+x_1x_2}x_1 \right).
	\end{equation}
	\item If we first replace the parameters by $a\rightarrow -\frac{1}{\alpha}$ and $b\rightarrow -\frac{1}{\beta}$, and then we replace the dependent variables as $u_1\rightarrow -\frac{\alpha}{\beta}u_1$ and $u_2\rightarrow -\frac{\beta}{\alpha}u_2$, the above map can be written as
	\begin{equation}
	(x_1,x_2)\overset{Y_{a,b}}{\longrightarrow } (u_1(x_1,x_2),u_2(x_1,x_2))=\left(\frac{1+\alpha x_1x_2}{1+\beta x_1x_2}x_2,\frac{1+\beta x_1x_2}{1+\alpha x_1x_2}x_1 \right),
	\end{equation}
	which is the fourth map of the $H$-list \cite{PSTV}.
\end{enumerate}
\end{proof}


\subsection{A non-involutive Yang-Baxter map}
Now, we change $(p,\tilde{q},f;c_1,c_2)\rightarrow (x_1,x_2,x_3;1,1)$ in (\ref{Z2-M-gen}), we define the following map
\begin{equation} \label{Z2-M-gen-Dar}
M(\textbf{x})= \frac{x_1}{\lambda-1} \left(\begin{array}{cc}
x_3 & -x_2 x_3 \\ 1 & -x_2 \end{array} \right) -
\frac{x_1}{\lambda+1}\left(\begin{array}{cc} x_3 & x_2 x_3 \\ -1 & -x_2
\end{array} \right) + \left(\begin{array}{cc} 1 &0 \\ 0 & 1 \end{array}
\right),
\end{equation}
where $\textbf{x}:=(x_1, x_2,x_3)$ and substitute it to the Lax equation
\begin{equation}\label{L-DNLS_gen-2}
M(\textbf{u})M(\textbf{v})=M(\textbf{y})M(\textbf{x}).
\end{equation}

The Lax equation \eqref{L-DNLS_gen-2} implies a system of polynomial equations in the respective arguments which defines the algebraic variety. Here, the algebraic variety is union of two six-dimensional components. The first one is obvious from refactorisation \eqref{L-DNLS_gen} and it corresponds to the permutation map
\begin{equation*}
\textbf{u}\mapsto \textbf{y}, \qquad \textbf{v} \mapsto \textbf{x},
\end{equation*}
while the second one can be represented as the following birational six-dimensional YB map
\begin{subequations}\label{YB-DNLS_gen-2}
\begin{align}
x_1\mapsto u_1 &=\frac{x_1(x_1-x_3-2x_1x_2x_3)+y_1(y_2-x_2+2x_1x_2x_3)}{x_2(1-2x_1x_3)-y_3(1+2y_1y_2)},\\
x_2\mapsto u_2 &=\frac{[x_2(1-2x_1x_3)-y_3(1+2y_1y_2)][x_1x_2(x_2-x_3-2x_1x_2x_3)+y_1y_3(x_2-y_3-2x_1x_2y_2)]}{\left[x_1x_3-x_1x_2(1-2x_1x_2)+y_1y_2-x_2y_1(1-2x_1x_3)\right]\left[y_3-x_2+2x_2y_3(x_1+y_1)\right]},\\
x_3\mapsto u_3 &=y_3,\\
y_1\mapsto v_1 &=\frac{x_1(x_3-y_3-2y_1y_2y_3)+y_1(y_2-y_3-2y_1y_2y_3)}{x_2(1-2x_1x_3)-y_3(1+2y_1y_2)},\\
y_2\mapsto v_2 &=x_2,\\
y_3\mapsto v_3 &=\frac{[x_2(1-2x_1x_3)-y_3(1+2y_1y_2)][x_1x_2(x_2-x_3-2x_1x_2x_3)+y_1y_3(x_2-y_3-2x_1x_2y_2)]}{[x_1(x_3-y_3-2y_1y_2y_3)+y_1(y_2-y_3-2y_1y_2y_3)][x_2-y_3-2x_2y_3(x_1+y_1)]},
\end{align}
\end{subequations}
which is non-involutive.

We now make use of the first integral \eqref{z2-fi}.

\begin{proposition}
Map \eqref{YB-DNLS_gen-2} can be restricted to the following four-dimensional, non-involutive, parametric Yang-Baxter map, $Y\in\End(A_a\times B_b)$ on the invariant leaves $A_a:=\left\{\right (x_1,x_2,x_3)\in K^3;(1+2x_1x_2)$ $\left (1+2x_1x_3)=a\right\}$ and $B_b:=\left\{(y_1,y_2)\in K^3;(1+2y_1y_2)(1+2y_1y_3)=b\right\}$.
	\begin{subequations}\label{YB-DNLS_gen-2-4D}
	\footnotesize
\begin{align}
x_1\mapsto u_1 &=y_1+\frac{(a-b)(1+2x_1x_2)}{b+2(a-2)x_2y_1+2bx_1x_2-(1+2x_1x_2)(1-2y_1y_2)}y_1,\\
x_2\mapsto u_2 &=\frac{[(a-b)x_2+(b-1+2y_1y_2)(1+2x_1x_2)+2(a-2)x_2y_1y_2][b+2(a-2)x_2y_1+2bx_1x_2-(1+2x_1x_2)(1-2y_1y_2)]}{[a+2(a-2)x_2y_1+2ax_1x_2-(1+2x_1x_2)(1-2y_1y_2)][b+2(b-1)x_1x_2+2(b-2)x_2y_1-(1+2x_1x_2)(1-2y_1y_2)]},\\
y_1\mapsto v_1 &=x_1-\frac{(a-b)(1+2x_1x_2)}{b+2(a-2)x_2y_1+2bx_1x_2-(1+2x_1x_2)(1-2y_1y_2)}y_1,\\
y_2\mapsto v_2 &=x_2,
\end{align}
\end{subequations}
Moreover, map \eqref{YB-DNLS_gen-2} has the following invariants
\begin{align*}
I_1 &=x_1+y_1,\\
I_2 &=(1+2x_1x_2)(1+2y_1y_2).
\end{align*}\normalsize
\end{proposition}
\begin{proof}
The existence of the first integral \eqref{z2-fi} indicates that the corresponding quantities $\Psi:=(1+2x_1x_2)(2+x_1x_3)$ and $\Omega:=(1+2y_1y_2)(2+y_1y_3)$ are invariants of map \eqref{YB-DNLS_gen-2}; this can be verified by direct computation. Thus, we set $\Psi=a$, $\Omega=b$ and we eliminate $x_3$ and $y_3$ in \eqref{YB-DNLS_gen-2} by
\begin{equation}
x_3=\frac{a-1+2x_1x_2}{2x_1(1+2x_1x_2)} \qquad y_3=\frac{b-1+2y_1y_2}{2y_1(1+2y_1y_2)}.
\end{equation}
The resulting map can be expressed as \eqref{YB-DNLS_gen-2-4D}.

The invariance under $I_1$ and $I_2$ can be readily verified by straightforward calculation.
\end{proof}

\begin{remark}\normalfont
The Lax representation of map \eqref{YB-DNLS_gen-2-4D} is given by \eqref{L-DNLS_gen-2}, where
\begin{equation*} 
M(\textbf{x})=\frac{1}{\lambda-1} \left(\begin{array}{cc}
\frac{1}{2}\frac{a-1+2x_1x_2}{1+2x_1x_2} & \frac{-1}{2}\frac{a-1+2x_1x_2}{1+2x_1x_2}x_2 \\ x_1 & -x_1x_2 \end{array} \right) -
\frac{1}{\lambda+1}\left(\begin{array}{cc}
\frac{1}{2}\frac{a-1+2x_1x_2}{1+2x_1x_2} & \frac{1}{2}\frac{a-1+2x_1x_2}{1+2x_1x_2}x_2 \\ -x_1 & -x_1x_2 \end{array} \right) + \left(\begin{array}{cc} 1 &0 \\ 0 & 1 \end{array}
\right).
\end{equation*}
The invariants $I_1$ and $I_2$ of map \eqref{YB-DNLS_gen-2-4D} are generated by $\tr(M(\textbf{y})M(\textbf{x}))$.
\end{remark}

\begin{remark}\normalfont
In the derivation of map \eqref{YB-DNLS_gen-2-4D}, we chose $c_1=c_2=1$ in \eqref{Z2-M-gen}. The other option would be to rescale one of $c_i$, $i=1,2$ to 1, since $|c_1|^2+|c_2|^2 \ne 0$. However, the latter choice leads to the permutation map, which is a trivial solution of the Yang-Baxter equation.
\end{remark}


\section{Solutions to the Entwining Yang-Baxter equation}
In this section, we show that matrix refactorisation problems of Darboux matrices give rise to solutions of the parametric entwining Yang-Baxter equation. We also prove that the derived maps are completely integrable.


\subsection{NLS case}
We use the Darboux matrices of the NLS equation, \eqref{M-NLS} and \eqref{M-degen}, and we define the following Lax matrices
\begin{subequations}\label{LaxNLSL123}
\begin{align}
L_1(\textbf{x};a)&:=\lambda \left(\begin{array}{cc}1 & 0\\0 & 0\end{array}\right)+\left(\begin{array}{cc}a+x_1x_2 & x_1\\x_2 & 1 \end{array}\right),\\
L_2(\textbf{x};a)&:=\lambda \left(\begin{array}{cc}1 & 0\\0 & 0\end{array}\right)+\left(\begin{array}{cc}x_2 & x_1\\\frac{a}{x_1} & 0 \end{array}\right),\\
L_3(\textbf{x};a)&\equiv L_1(\textbf{x};a),
\end{align}
\end{subequations}
where $\textbf{x}:=(x_1,x_2)$.

\begin{theorem}
The system of maps ($S_{a,b}$, $R_{a,b}$, $T_{a,b}$), given by
\begin{subequations}
\begin{align}
(\textbf{x},\textbf{y})\overset{S_{a,b}}{\longrightarrow } (\textbf{u}_{s},\textbf{v}_{s})&= \left(\frac{y_1+x_1(y_2-a)}{bx_1}y_1,\frac{b}{y_1},x_1,a+x_1x_2-\frac{y_1}{x_1}\right),\label{S-NLS}\\
(\textbf{x},\textbf{y})\overset{R_{a,b}}{\longrightarrow }(\textbf{u}_{r},\textbf{v}_{r})&=\left(y_1-\frac{a -b}{1+x_1y_2}x_1,y_2,x_1,x_2+\frac{a -b}{1+x_1y_2}y_2\right),\label{Adler-Yamilov}\\
(\textbf{x},\textbf{y})\overset{T_{a,b}}{\longrightarrow } (\textbf{u}_{t},\textbf{v}_{t})&= \left(\frac{a}{y_2},b+y_1y_2-\frac{a}{x_1y_2},x_1,\frac{a+x_1y_2(x_2-b)}{x_1^2y_2}\right),\label{T-NLS}
\end{align}
\end{subequations}
 is a birational, parametric, entwining Yang-Baxter map with strong Lax triple ($L_1,L_2,L_3$) given by \eqref{LaxNLSL123}.
\end{theorem}
\begin{proof}
We first consider the following matrix refactorisation problem
\begin{equation*}
L_1(\textbf{u}_s;a)L_2(\textbf{v}_s;b)=L_2(\textbf{y};b)L_1(\textbf{x};a).
\end{equation*}
This implies a system of polynomial equations in $\textbf{u}_s, \textbf{v}_s,\textbf{x},\textbf{y}$ which admits a unique solution for $\textbf{u}_s$ and $\textbf{v}_s$, namely the map $(u_s,v_s)=S_{a,b}(x,y)$ given by \eqref{S-NLS}. Moreover, from the following refactorisation problem
\begin{equation}\label{L1L3-NLS}
L_1(\textbf{u}_r;a)L_3(\textbf{v}_r;b)=L_3(\textbf{y};b)L_1(\textbf{x};a),
\end{equation}
follows a system of polynomial equations in $\textbf{u}_r, \textbf{v}_r,\textbf{x},\textbf{y}$ which uniquely defines the map $(u_r,v_r)=R_{a,b}(x,y)$ given by \eqref{Adler-Yamilov}. Finally, the refactorisation problem between matrices $L_2$ and $L_3$,
\begin{equation*}
L_2(\textbf{u}_t;a)L_3(\textbf{v}_t;b)=L_3(\textbf{y};b)L_2(\textbf{x};a),
\end{equation*}
is equivalent to the map $(u_t,v_t)=T_{a,b}(x,y)$ given by \eqref{T-NLS}.

It can be verified that the system of polynomial equations in $\textbf{u},\textbf{v},\textbf{w},\textbf{x},\textbf{y},\textbf{z}$ that follows from the matrix trifactorisation problem $L_1(\textbf{x};a)L_2(\textbf{y};b)L_3(\textbf{z},c)=L_1(\textbf{x}';a)L_2(\textbf{y}';b)
L_3(\textbf{z}',c)$ admits only the trivial solution $\textbf{x}=\textbf{x}', \textbf{y}=\textbf{y}', \textbf{z}= \textbf{z}'$. Therefore, the system of maps (\eqref{S-NLS},\eqref{Adler-Yamilov}, \eqref{T-NLS}) satisfies the parametric entwining Yang-Baxter equation.

The birationality of $R$ follows from the properties of \eqref{L1L3-NLS}. For the maps $S$ and $T$, birationality can be verified by direct calculation. 
\end{proof}

\begin{remark}\normalfont
Map \eqref{Adler-Yamilov} is the famous Adler-Yamilov Yang-Baxter map \cite{Adler-Yamilov,  Kouloukas, PT} which is completely integrable. 
\end{remark}

For the integrability of maps $S$ and $T$ we have the following.

\begin{proposition}
Map $S$, given by \eqref{S-NLS}, is Liouville integrable.
\end{proposition}

\begin{proof}
From the trace of the monodromy matrix, $M(\textbf{x},\textbf{y})=L_3(\textbf{y};b)L_1(\textbf{x};a)$, we obtain the following invariants
\begin{equation}
I_1=x_1x_2+y_2, \qquad I_2=x_1x_2y_2+x_2y_1+ay_2+b\frac{x_1}{y_1}.
\end{equation}
It can be readily verified that these invariants are in involution with respect to the following Poisson bracket 
\begin{equation}
\label{eq:Poisson-1}
\left\{x_1,x_2\right\}=1,\quad \left\{y_1,y_2\right\}=y_1\qquad \text{and all the rest}
\qquad \left\{x_i,y_j\right\}=0.
\end{equation}
The rank of the associated Poisson matrix is 4 and it is being preserved by map $S$. 
\end{proof}

\begin{proposition}
Map $T$, given by \eqref{T-NLS}, is Liouville integrable.
\end{proposition}

\begin{proof}
The quantity $\tr (L_3(\textbf{y};b)L_1(\textbf{x};a))$ generates for map $T$ the following invariants
\begin{equation}
I_1=x_2+y_1y_2, \qquad I_2=x_2y_1y_2+x_1y_2+bx_2+a\frac{y_1}{x_1},
\end{equation}
which are in involution with respect to the following Poisson bracket
\begin{equation}
\label{eq:Poisson-2}
\left\{x_1,x_2\right\}=x_1,\quad \left\{y_1,y_2\right\}=1\qquad \text{and all the rest}
\qquad \left\{x_i,y_j\right\}=0,
\end{equation}
and all the rest of conditions for complete integrability are satisfied.
\end{proof}

Since the Poisson matrices associated to \eqref{eq:Poisson-1} and \eqref{eq:Poisson-2} have full rank they define symplectic 2-forms making the corresponding maps $S$ and $T$ birational symplectic maps. The symplectic 2-forms
$$
\omega_S=dx_1\wedge dx_2+\frac{1}{y_1}dy_1\wedge dy_2,
$$ 
and
$$
\omega_T=\frac{1}{x_1}dx_1\wedge dx_2+dy_1\wedge dy_2,
$$
are associated with \eqref{eq:Poisson-1} and \eqref{eq:Poisson-2}, respectively. It can be verified that
$$
S^*\omega_S=\omega_S \quad T^*\omega_T=\omega_T,
$$
and thus it follows that the maps $S$ and $T$ are also measure preserving maps since they preserve the volume forms $\omega_S\wedge\omega_S$ and $\omega_T\wedge\omega_T$, respectively.

\subsection{DNLS case}
Now, according to \eqref{DT-sl2-gen} and \eqref{DT-sl2-degen} we define a Lax triple. In particular, we set $(x_1,x_2):=(fp,f\tilde{q})$ in \eqref{DT-sl2-gen} and $(x_1,x_2):=(f,fp)$ in \eqref{DT-sl2-degen}
\begin{subequations}\label{LaxDNLSL123}
\begin{align}
L_1(\textbf{x};a)&:=\lambda^{2}\left(\begin{array}{cc} a+x_1x_2 & 0\\ 0 & 0\end{array}\right)+\lambda\left(\begin{array}{cc} 0 & x_1\\ x_2 & 0\end{array}\right)+\left(\begin{array}{cc} 1 & 0\\ 0 & 1 \end{array}\right),\\
L_2(\textbf{x};a)&:=\lambda^{2}\left(\begin{array}{cc} x_1 & 0\\ 0 & 0\end{array}\right)+\lambda\left(\begin{array}{cc} 0 & x_2\\ \frac{k}{x_2} & 0\end{array}\right)+\left(\begin{array}{cc} 1 & 0\\ 0 & 0 \end{array}\right)\\
L_3(\textbf{x};a)&\equiv L_1(\textbf{x};a),
\end{align}
\end{subequations}
where $\textbf{x}:=(x_1,x_2)$.

\begin{theorem}
The system of maps ($S_{a,b}$, $R_{a,b}$, $T_{a,b}$), given by
\begin{subequations}
\begin{align}
(\textbf{x},\textbf{y})\overset{S_{a,b}}{\longrightarrow } (\textbf{u}_{s},\textbf{v}_{s})&= \left(\frac{x_1y_1-a(x_1+y_2)}{bx_1}y_2,\frac{bx_1}{y_2(x_1+y_2)},\frac{(x_1+y_2)(a+x_1x_2)}{x_1},x_1+y_2\right),\label{S-DNLS}\\
(\textbf{x},\textbf{y})\overset{R_{a,b}}{\longrightarrow }(\textbf{u}_{r},\textbf{v}_{r})&=\left(y_1+\frac{a-b}{a-x_1y_2}x_1,\frac{a-x_1y_2}{b-x_1y_2}y_2,\frac{b-x_1y_2}{a-x_1y_2}x_1,x_2+\frac{b-a}{b-x_1y_2}y_2\right),\label{YB-DNLS}\\
(\textbf{x},\textbf{y})\overset{T_{a,b}}{\longrightarrow } (\textbf{u}_{t},\textbf{v}_{t})&= \left(\frac{(a+x_2y_2)(b+y_1y_2)}{x_2y_2},\frac{ax_2}{a+x_2y_2},\frac{x_2^2y_2}{a+x_2y_2},\frac{x_1x_2y_2-b(a+x_2y_2)}{x_2^2y_2}\right),\label{T-DNLS}
\end{align}
\end{subequations}
 is a birational parametric, entwining Yang-Baxter map with strong Lax triple ($L_1,L_2,L_3$) given by \eqref{LaxDNLSL123}.
\end{theorem}
\begin{proof}
The following matrix refactorisation problem
\begin{equation}\label{L1L2-DNLS}
L_1(\textbf{u}_s;a)L_2(\textbf{v}_s;b)=L_2(\textbf{y};b)L_1(\textbf{x};a),
\end{equation}
implies a system of polynomial equations in $\textbf{u}_s, \textbf{v}_s,\textbf{x},\textbf{y}$ which admits a unique solution for $\textbf{u}_s$ and $\textbf{v}_s$, namely the map $(u_s,v_s)=S_{a,b}(x,y)$ given by \eqref{S-DNLS}. We, then, consider the following refactorisation problem
\begin{equation}\label{L1L3-DNLS}
L_1(\textbf{u}_r;a)L_3(\textbf{v}_r;b)=L_3(\textbf{y};b)L_1(\textbf{x};a).
\end{equation}
This implies a system of polynomial equations in $\textbf{u}_r, \textbf{v}_r,\textbf{x},\textbf{y}$ which uniquely defines the map $(u_r,v_r)=R_{a,b}(x,y)$ given by \eqref{YB-DNLS}. Finally, the refactorisation problem between matrices $L_2$ and $L_3$,
\begin{equation}\label{L2L3-DNLS}
L_2(\textbf{u}_t;a)L_3(\textbf{v}_t;b)=L_3(\textbf{y};b)L_2(\textbf{x};a),
\end{equation}
is equivalent to the map $(u_t,v_t)=T_{a,b}(x,y)$ given by \eqref{T-DNLS}.

It can be verified that the system of polynomial equations in $\textbf{u},\textbf{v},\textbf{w},\textbf{x},\textbf{y},\textbf{z}$ that follows from the matrix trifactorisation problem $L_1(\textbf{x};a)L_2(\textbf{y};b)L_3(\textbf{z},c)=L_1(\textbf{x}';a)L_2(\textbf{y}';b)
L_3(\textbf{z}',c)$ admits only the trivial solution $\textbf{x}=\textbf{x}',\textbf{y}=\textbf{y}',\textbf{z}=\textbf{z}'$. Therefore, the system of maps (\eqref{S-DNLS},\eqref{YB-DNLS}, \eqref{T-DNLS}) satisfies the parametric entwining Yang-Baxter equation.

From the properties of \eqref{L1L3-DNLS} it follows that map $R$ is birational. The birationality of maps $S$ and $T$ can be directly verified. 
\end{proof}

\begin{remark}\normalfont
Map \eqref{YB-DNLS} was derived in \cite{Sokor-Sasha} and it is, itself, a completely integrable, parametric Yang-Baxter map.
\end{remark}

\begin{proposition} 
Map $S$, given by \eqref{S-DNLS}, is completely integrable.
\end{proposition}

\begin{proof}
The trace of the monodromy matrix, $M(\textbf{x},\textbf{y})=L_3(\textbf{y};b)L_1(\textbf{x};a)$, implies the following invariants
\begin{equation*}
I_1=y_1 (x_1x_2+a), \qquad I_2=x_1x_2+y_1+x_2y_2+b\frac{x_1}{y_2}.
\end{equation*}
These invariants are in involution with respect to the following Poisson bracket 
\begin{equation}\label{Poisson-SDNLS}
\left\{x_1,x_2\right\}=x_1x_2+a,\quad \left\{x_1,y_2\right\}=x_1y_2,\quad \left\{x_2,y_2\right\}=-x_2y_2,\quad\left\{y_1,y_2\right\}=-y_1y_2,
\end{equation}
and the rest $\left\{x_1,y_1\right\}=\left\{x_2,y_1\right\}=0$. The rank of the associated Poisson matrix is 4 and it is being preserved by map $S$. 
\end{proof}

\begin{proposition}
Map $T$, given by \eqref{T-DNLS}, is completely integrable.
\end{proposition}
\begin{proof}
The trace of the monodromy matrix $M(\textbf{x},\textbf{y})=L_3(\textbf{y};b)L_1(\textbf{x};a)$ generates the following invariants
\begin{equation*}
I_1=x_1 (y_1y_2+b), \qquad I_2=x_1+y_1y_2+x_2y_2+b\frac{y_1}{x_2},
\end{equation*}
which are in involution with respect to the following Poisson bracket 
\begin{equation}\label{Poisson-TDNLS}
\left\{x_1,x_2\right\}=-x_1x_2,\quad \left\{x_2,y_1\right\}=x_2y_1,\quad \left\{x_2,y_2\right\}=-x_2y_2,\quad\left\{y_1,y_2\right\}=y_1y_2+\beta,
\end{equation}
and the rest $\left\{x_1,y_1\right\}=\left\{x_1,y_2\right\}=0$. The rank of the associated Poisson matrix is 4 and it is being preserved by map $T$. 
\end{proof}

Here, maps $S$ and $T$ are symplectic, since the Poisson matrices associated to the brackets \eqref{Poisson-SDNLS} and \eqref{Poisson-TDNLS} are invertible. The corresponding symplectic 2-forms are given by 
$$
\omega_S = -\frac{1}{x_1x_2+a}dx_1\wedge dx_2+\frac{x_2}{x_1x_2+ay_1}dx_1\wedge dy_1+\frac{x_1}{x_1x_2y_1+ay_1}dx_2\wedge dy_1+\frac{1}{y_1y_2}dy_1\wedge dy_2,
$$
and
$$
\omega_T = -\frac{1}{x_1x_2}dx_1\wedge dx_2+\frac{y_2}{x_1(y_1y_2+b)}dx_1\wedge dy_1+\frac{y_1}{x_1(y_1y_2+b)}dx_2\wedge dy_2-\frac{1}{y_1y_2+b}dx_2\wedge dy_1,
$$
related to $S$ and $T$, respectively.


\section{Conclusions}
In the first part of the paper, we employed Darboux matrices related to the DNLS equation and we derived an involutive and a non-involutive Yang-Baxter map, namely maps \eqref{YB-DNLS_gen} and \eqref{YB-DNLS_gen-2}. We showed that \eqref{YB-DNLS_gen} can be restricted to the famous $H_{IV}$ map on symplectic leaves. Moreover, map \eqref{YB-DNLS_gen-2} can be restricted to a novel, non-involutive, parametric Yang-Baxter map on invariant leaves. We conjecture that it is completely integrable, namely its invariants $I_1$ and $I_2$ are in involution with respect to a Poisson bracket.

In the second part of the paper, we used Darboux matrices associated to the NLS and the DNLS equation together with their degenerated versions as Lax triples to derive Entwining Yang-Baxter maps. In particular, we constructed the $(S,R,T)$ maps $\eqref{S-NLS},\eqref{Adler-Yamilov},\eqref{T-NLS}$ and $\eqref{S-DNLS},\eqref{YB-DNLS},\eqref{T-DNLS}$ and we showed that they are completely integrable in the Liouville sense.

Our results can be extended in the following ways:
\begin{enumerate}
\item Construction of solutions to the Entwining Tetrahedron equation;
\item Derivation of auto-B\"acklund transformations from entwining Yang-Baxter maps;
\item Vector generalisation of the derived maps;
\item Extensions of the results to the noncommutative case.
\end{enumerate}

Regarding 1., to the best of out knowledge, there are no known rational solutions to the  Tetrahedron equation with parameters. On the other hand, the relation between matrix trifactorisation problems and the Tetrahedron equation is already studied in \cite{Kashaev}. We believe that Darboux transformations can be employed to construct birational solutions to the parametric entwining Tetrahedron equation, namely
$$
Q^{123} \circ S^{145} \circ R^{246} \circ T^{356} = T^{356} \circ R^{246} \circ S^{145}\circ Q^{123},
$$
or a modification of it.

In point 3., it is meant that similar results to this paper can be obtained if one replaces the variables in the associate Lax operators by vectors. However, the trace of the monodromy matrix in this case may not give sufficient number of invariants in order to claim Liouville integrability, and invariants need to be found using other methods.

Concerning 2., recall that the variables of the entwining Yang-Baxter maps presented in this paper are potential functions-solutions to NLS type of Integrable PDEs. Demanding that the map preserves these solutions, we can obtain some relations between them. It should be examined in which cases these are indeed B\"acklund transformations.

Finally, a lot of interest has been drawn in the construction of noncommutative versions of discrete integrable systems and Yang-Baxter maps (indicatively we refer to \cite{Bobenko-Suris,Doliwa}). Moreover, the first steps in extending the theory of Yang-Baxter maps to the noncommutative Grassmann case were already made \cite{GKM, Georgi-Sasha, Sokor-Sasha-2, Sokor-Kouloukas}. All the results of this paper can be extended to the Grassmann case.

\section*{Acknowledgements}
The authors would like to thank T.E. Kouloukas for numerous discussions on this topic. S.K.R. acknowledges that this work was carried out within the framework of the State Programme of the Ministry of
Education and Science of the Russian Federation, Project No. 1.13560.2019/13.1. This work was initiated during SKR's visit to the University of Leeds on a LMS short visit grant (Ref. 21717). The work of G.P. was supported by an EPSRC grant EP/P012655/1.

\end{document}